\documentclass[10pt,journal,compsoc]{IEEEtran}
\usepackage{amssymb} %为了使用正整数集

\usepackage{graphicx}
\usepackage{amsmath,amssymb,amsthm}

\newtheorem{definition}{Definition}[section] %使用定义，则需要自己加上环境
\newtheorem{theorem}{Theorem}[section]%using theory format
\newtheorem{lemma}{Lemma}[section]

\usepackage{multirow}
\usepackage{lipsum}
\usepackage{subfigure}
\usepackage{algorithm}
\usepackage{algorithmic}

\usepackage{array}
\usepackage{booktabs}
\usepackage{cite}

\ifCLASSINFOpdf
\else
\fi
\begin{document}
\title{Non-cooperative game approach for task offloading in edge clouds}

\author{Bo~Yang,~\IEEEmembership{Member,~IEEE},
        Zhiyong~Li,~\IEEEmembership{Member,~IEEE}, and Wenbin~Liu
         %and Keqin~Li,~\IEEEmembership{Fellow,~IEEE} % <-this % stops a space
\IEEEcompsocitemizethanks{\IEEEcompsocthanksitem Bo Yang and Wenbin Liu are with the College of Information and Management of Hunan University of Finance and Economics, Changsha, China, 410205.\protect\\
E-mail:\{bo\_yang, zhiyong.li\}@hnu.edu.cn.
\IEEEcompsocthanksitem  Bo Yang and Zhiyong Li are with the College of Computer Science and Electronic Engineering of Hunan University, National Supercomputing Center in Changsha, Key Laboratory for Embedded and Network Computing of Hunan Province, Changsha, China, 410082.
%\IEEEcompsocthanksitem Keqin~Li is also with the Department of Computer Science, State University of New York, New Paltz, New York 12561.
 }% <-this % stops a space
\thanks{Manuscript received xx xx, 2017; revised July xx, xxxx.}}


\markboth{Journal of \LaTeX\ Class Files,~Vol.~ , No.~ , September~ }%
{Shell \MakeLowercase{\textit{et al.}}: Bare Advanced Demo of IEEEtran.cls for Journals}

%\IEEEpubid{\makebox[\columnwidth]{\hfill 0000--0000/00/\$00.00~\copyright~2014 IEEE}%


\IEEEtitleabstractindextext{%
\begin{abstract}
Task offloading provides a promising way to enhance the capability of the mobile terminal (also called terminal user) that is distributed on network edge and communicates edge clouds with wireless. Generally, there are multiple edge cloud nodes with distinct processing capability in a geographic area, which can offer computing service for various terminal users. Furthermore, the terminal users are competitive and selfish, i.e., each user takes into account only maximizing her own profit, while conducting task offloading strategies. In this paper, we focus on the resource management optimization for edge clouds, and formulate the problem of resource competition among terminal users as a non-cooperative game, in which the terminal user who acts as the player always pursues the minimization of the expected response time for her tasks by optimizing allocation strategies. We present the utility function of the user with queuing theory, and then prove the existence of Nash equilibrium for the formulated game. Using the concept of Nash bargaining solution to calculate the optimal task offloading scheme for the user, we propose a distributed task offloading algorithm with low computation complexity. The results of simulated experiments demonstrate that our method can quickly reach the Nash equilibrium point, and deliver satisfying performance at the expected response time of the user's tasks.
\end{abstract}

}

\maketitle

% is not selected <OR> if conference mode is selected - because compsoc
% conference papers position the abstract like regular (non-compsoc)
\IEEEdisplaynontitleabstractindextext

% \fi
\IEEEpeerreviewmaketitle

\ifCLASSOPTIONcompsoc
\IEEEraisesectionheading{\section{Introduction}\label{sec:introduction}}
\else
\section{Introduction}
\label{sec:introduction}
\fi

\subsection{Motivation}

\IEEEPARstart{W}ith the development of high-powered mobile terminals and 5G communication technologies, mobile computing has received extensive attention. In this scenario, the applications of the mobile user often are offloaded to the cloud due to the capability constraint of the terminal. For example, to quickly respond the computing-intensive speech recognition, Google Voice Search and Apple Siri always request computing resources from cloud data centers\cite{Liu2013Gearing}. However, they might suffer the risk of uncertain delay caused by the network while the user's tasks are transmitted to the remote cloud data center via multiple network nodes. To address the problem, the edge cloud computing (ECC) model has been emerged. Unlike traditional cloud computing models, it transfers a number of computing resources from the cloud to network edges closest to the user, allowing time-sensitive tasks to be processed nearby and avoiding the uncertainty of transmission delay \cite{Nguyen2017Resource}.

Due to the advantage of the proximate access for the mobile user, ECC has been widely agreed to be a key technology for the mobile computing. However, compared to traditional large-scale cloud computing, an ECC node is usually equipped with only small or medium size computing resources because of the limit of the physical space. Hence, it cannot meet all users when the number of user requests is relatively large. In fact, a mobile user can obtain identical service from various ECC nodes in a physical field, whereas a ECC node can also provide service for multiple users. Accordingly, it is inevitable to consider the load balancing among ECC nodes, so that the system can provide better experience for the users. The existing researches for ECC mainly focused on the task offloading between the mobile device and the edge cloud node, such as user application partitioning and offloading \cite{Wang2017Online, Huang2012A}, access control\cite{Chen2015, Hoang2012Optimal} and so on. Although these researches in these areas are very important, they did not taken into account both of the resource competition among users and the queuing delay of user tasks in ECC systems.

Without loss of generality, most of the mobile users in the ECC system are time-sensitive and independent from each other. The objective of the user is to minimize the expected response time of her tasks while she makes decisions for task offloading.
However, to capture this objective, there are several significant points that need further consideration.  For examples, the users of ECC systems are free to act independently in a selfish manner, such that achieving satisfactory experience for users is difficult due to no central authority in the system; the response time of the user task not only refers to its running time and transmission delay, but also includes its queuing delay on the ECC node. Notably, the processing capability of the mobile terminal is far lower than that of the edge cloud node in fact. Nevertheless, the mobile users can also execute their own part of tasks locally when external computing resources provided by the edge cloud node cannot meet their requirements\cite{Wang2017Online,Chen2015}. However, currently task offloading strategies completely considering above scenarios have not been well studied. Therefore, it is still a challenge to design an efficient, stable and distributed method of task offloading in the ECC system.

\subsection{Our contributions}

Game theory can provide a natural paradigm to design decentralized mechanisms \cite{Myerson2013}, which can help obtain an in-depth analytical understanding of the task allocation problem of ECCs. Accordingly, in this study, we introduce the game theoretic approach to address the aforementioned issues. In our model, ECC users, who act as players, compete with one another in sharing resources distributed in multiple ECC nodes in a physical field. To minimize the average response time of tasks, they can decide which ECC nodes will process their requirements by observing the running scenario, such as the available computing resources, the transmitting delay and the service levels of their requirements. Our main contributions are summarized as follows.
	
Our schemes consider the expected response time that includes both of the queueing time and the transmitting delay for tasks. We establish a queueing model to characterize the computing node in the ECC system. The task offloading problem of each user is viewed as a centralized problem in our schemes, because each mobile user independently acts with a selfish manner. Furthermore, to ease the difficulty of solving the problem, Nash bargaining solution (NBS) method is adopted to calculate the optimal strategy for the mobile user.

Additionally, we propose a non-cooperative game framework for the mobile edge cloud system, in which each mobile user can selfishly minimize her payoff by making optimal strategy for her task offloading with the NBS method. We prove the existence of the Nash equilibrium of the proposed game. Moreover, corresponding algorithms are given to find the best response of the mobile user, and can converge to an efficient equilibrium after several iterations. The experimental results demonstrate that our approaches are effective and efficient.

The rest of this paper is organized as follows. We first discuss related works in Section 2, and then introduce the system model in Section 3. We provide the formulation of the resources allocation problem and present the detailed description of the proposed algorithms in Sections 4 and 5, respectively. We discuss the simulation results that demonstrate the effectiveness of our approach in Section 6. Finally, the conclusions are drawn in Section 7.

\section{Related Work}

In recent years, due to the improvement of communication technology and mobile terminals, more resources of cloud computing have been transferred to the network edge closed to customers. The existing researches on ECCs are mainly focus on the problem model and resource management of the system \cite{Mao2017}. In this section, we will discuss the associated work from following aspects: the task partition, and the task offloading among multiple users in ECCs.

Generally, a mobile computing task is consist of multiple procedures, such as the augment reality (AR) task, which has five critical procedures: the video collection, the tracker, the mapper, the object recognizer, and the renderer. Among these procedures, the computation-intensive procedures, i.e., the tracker, mapper and object recognizer can be offloaded for cloud execution, others can be performed locally \cite{Mao2017}. In this way, the mobile user can enjoy various benefits from ECCs using distinct partition algorithms. Hence, a series of task offloading approaches are proposed to optimize the system performance\cite{Wang2017Online,Yang2015Multi,Xiang2014Ready,Chang2014Bringing}.

Ref\cite{Wang2017Online} focused on the upper bound of the system performance, and investigated the problem of the task partition and placement, in which the user's task and the physical computing system are mapped into two types of nodes that are labelled as the requirement and the available resource in a graph respectively, and then an approximate online algorithm was proposed to solve the optimal matching for the two types of nodes in this graph. Considering the scenario where the available resources are constrained in ECCs, authors of Ref\cite{Yang2015Multi} investigated the task partition methods that minimize the expected time of user tasks. Xiang \emph{et~al.} proposed a consolidated method to coordinate the task partition while conducting the task offloading strategy. In their work, saving energy is realized by reducing the time of the net interface working at high power\cite{Xiang2014Ready}. Chang \emph{et al.} studied the task placement approach for the application of Internet of Thing in the edge cloud, where the time-sensitive components of the task are  executed on the edge cloud node by splitting the task seamlessly\cite{Chang2014Bringing}.

Above investigations primarily focused on minimizing energy consumption and time delay, or considering the tradeoff of them. In fact, the multi-step task partitioning and placement in ECCs still encounter numerous challenges even if it is a simple task partitioning and placement. Moreover, these investigations mostly used heuristic algorithms to solve such problems\cite{Taleb2017}. It is well known that these methods generally have high computing complexity, and hence can not guarantee the performance. Furthermore, above studies neglected the task congestion or queuing delay, i.e., the user task arriving on an edge computing node may have to wait for a certain time before it can be executed. It would significantly harm the user's experience if the waiting time is too long.

For the scenario of multiple users, the popular schemes employed by existing studies for the resource management in ECCs can be classified into two types: centralization and decentralization. In these investigations\cite{Chen2016Joint,You2017Energy,Chen2017Joint} with centralized strategies, the edge cloud node has the completed information of all users, such as the resource request of the user, makes decisions for resource allocation and then sends them to users. Ref\cite{You2017Energy} studied resource allocation for a multiuser mobile edge cloud system, where the optimal resource allocation is formulated as a convex optimization problem for minimizing the weighted sum mobile energy consumption under the constraint on computation latency. Chen {\em et al.} aimed to minimize the overall cost of energy, computation and delay for all users, and formulated the optimization of the offloading decision as a non-convex quadratically constrained quadratic program\cite{Chen2016Joint}. Their work is further extended in the literature\cite{Chen2017Joint}, e.g., the computation resource allocation and processing cost were taken into account.

For the investigations with decentralization approaches, Sardellitti {\em et al.} formulated the offloading problem as a joint optimization of the radio resources and the computational resources, and provided a distributed resource scheduling algorithm using successive convex approximation technique to minimize the overall users' energy consumption, while meeting latency constraints\cite{Sardellitti2015Joint}. Similarly, authors of literature\cite{Guo2016Energy} proposed a distributed resource scheduling algorithm to reduce energy consumption and shorten application completion time in mobile edge clouds. However, the above studies did not considered the queuing delay and the user selfishness in the scenario of multiple users.

Moreover, the non-cooperative game approach, as a type of decentralized method, has obtained extensive attentions in current mobile edge computing fields. Ref\cite{Chen2015,Chen2016Efficient} assumed that the user tasks can run on the local terminal or the edge cloud, modelled the competing traffic channel among multiple users as a non-cooperative game, and proposed a task offloading approach being able to reach Nash equilibrium (NE). Cardellini et.al. investigated a scenario in which multiple non-cooperative users share the limited computing resources of close-by cloudlets in a fixed field and can selfishly determine to assign their computations to any of the three tiers, e.g., a local tier of mobile terminals, a middle tier of ECCs, and a remote tier of cloud servers\cite{Cardellini2016A}. Nevertheless, their work is for a single mobile edge cloud, and ignored a fact that there are multiple mobile edge nodes that can simultaneously offer services for users in a physical area. Considering a scenario that exists multiple users competing sharing computation resources provisioned by multiple ECCs in a field, Li {\em et al.} employed the queuing model to characterize the mobile edge cloud server, and proposed a non-cooperative game algorithm to find the optimal computation offloading strategy for the mobile users and multiple ECCs\cite{keqinli2018}. Although similar technologies are used in our work, the difference from their work is that we combine the Nash Bargain Solution (NBS) and the concept of non-cooperative game to calculate the optimal allocation scheme of the user, such that the analytic solution of the task allocation for each user can be quickly obtained.

\section{System  Model}
In this section, we introduce the models for the mobile edge cloud and user respectively. For the convenience of the readers, the major notations used in this paper are listed in Table 1.

\begin{table}[htbp] %\footnotesize
\centering  % 表居中
\caption{Notations} %\label{tab1}
\begin{tabular}{lp{7cm}}  %{lccc} 表示各列元素对齐方式，left-l,right-r,center-c
\toprule[1.0pt]
Symbol &Meaning\\
\hline  % \hline 在此行下面画一横线
$i$ &Subscript of the mobile user or terminal \\
$MU_i$ &Mobile user or terminal $i$\\
$j$ &Subscript of the computing node(e.g., edge cloud node, local terminal)\\
$\mathcal{M}$ &Set of the edge cloud node\\
$m$& The number of edge cloud nodes\\
$\mathcal{N}$ &Set of the mobile user\\
$ECN_j$ &Edge cloud node $j$\\
$\hat{\mu_i}$ &Available processing rate of $MU_i$\\
$\hat{\lambda}_i$ &Task arrival rate of $MU_i$\\
$T_i^{l}$ &Average response time of the tasks deployed on local computing for $MU_i$\\
$T_i^{e}$ &Average response time of the tasks deployed on ECNs for $MU_i$\\
$T_i$ &Overall average response time of the tasks for $MU_i$\\
$\rho_i$ &Task allocation probability vector for $MU_i$\\
$\rho_{ij}$ & Probability assigning the task belonged to $MU_i$ to computing node~$j$\\
$\mu _j^i $&Processing rate that computing node $j$ can provision to $MU_i$\\
$\mathcal{H}_i$ &Set of available computing nodes for $MU_i$\\
$\rho_{-i}$ &Task allocation probability vector of the users excepted user $i$\\
$\mathcal{L}_{ij}$ &Mean delay of transmitting a task of $MU_i$ to computing node $j$\\
$\mu_j^i$ &Available processing rate for $MU_i$ on computing node $j$\\
$\mathcal{Q}$ &Joined strategy set of all players\\
$\mathcal{T}_i^{max}$ &Acceptable maximum value of the expected task response time  for $MU_i$\\
$\mu_{ij}^0$ &Initial processing rate specified by $MU_i$ for computing node $j$\\
\toprule[1.0pt]
\end{tabular}
\end{table}

\subsection{Model of the mobile edge cloud}

Similar to WiFi APs, in a given physical area, we assume that there exists a mobile edge cloud system $\mathcal{M}$, which has $m$ edge cloud nodes (ECN), $\mathcal{M}=\{ECN_1, ECN_2, \cdots, ECN_m\}$. An edge cloud node can be envisioned as a small data center such as Cloudlet, which is able to provision service for the users closed to it, avoiding the unacceptable delay yielded by the uncertain network influence while the user task is transmitted to remote clouds through multiple network nodes. Moreover, there are a number of mobile users who can access certain edge cloud nodes in the area. Let $\mathcal{N}=\{MU_1, MU_2, \cdots, MU_n\}$ represent the set of mobile users. Fig. \ref{figure5:1} describes the model of the edge cloud system.
\begin{figure}[!ht]
  \centering
  \includegraphics[scale=0.65]{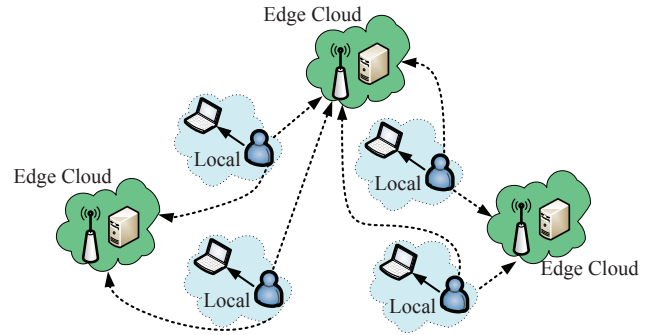}
  \caption{Mobile edge cloud model}
  \label{figure5:1}
\end{figure}

Assume that each mobile user generates tasks in terms of a Poisson process and independently of other mobile users. The task execution requirements (measured by the number of instructions to be executed) are i.i.d. exponential random variables \emph{r} with mean $\bar{r}$. An edge cloud node $j$ performs the user task with speed $s_j$ (measured by the number of instructions executed in one unit of time), and the execution times on the node are i.i.d. exponential random variables ${x_j} = r/{s_j}$ with mean ${\overline x _j} = {{{\overline r } \mathord{\left/
 {\vphantom {{\overline r } s}} \right.
 \kern-\nulldelimiterspace} s}_j}$. Hence, the average processing rate of the edge cloud node in one unit of time can be denoted by  ${\mu _j} = {{1 \mathord{\left/
 {\vphantom {1 {\overline x }}} \right.
 \kern-\nulldelimiterspace} {\overline x }}_j}$. Moreover, an edge cloud node maintains a queue with infinite capacity for waiting tasks, where the first-come-first-served queuing discipline is adopted. Such an edge cloud node can be modeled as an M/M/1 queuing system.

Notably, considering the fact that the service time of user tasks obeys general distributions, the M/G/1 queuing model should be more appropriate for them. However, there is still not a simple close-form solution for this queuing model, which can help us to understand the effect of the task size variability on the mean response time \cite{2013Performance}.

Let $\lambda_j$ denote the task arrival rate on edge cloud node $ECN_j$ with an average processing rate ${\mu _j}$. In terms of Little law\cite{2013Performance}, the average response time of $ECN_j$ can be calculated by
\begin{equation}\label{equ5:1}
  T_j = \frac{1}{{{\mu _j} - {\lambda _j}}}.
\end{equation}

For the system stability, the task arrival rate must be lower than the processing rate in an ECN. Hence, the following inequality should be hold.
\begin{equation}\label{equ5:2}
  \lambda _j  <\mu _j,\;\forall j \in \mathcal{M}.
\end{equation}

Otherwise, the queue at the ECN will build up to infinity and the expected response time for the user task will be infinite.

\subsection{Model of the edge cloud user}

Generally, the edge cloud user has a certain computing capability, which is far lower than that of the ECN, and can also offload her tasks to closed ECNs or remote cloud servers by the wireless network, while having to tackle the tasks with realtime requirements. Without loss of generality, in order to maximize her own profit, an edge cloud user may distribute strategically her tasks to suitable points (such as closed ECNs and her own intelligent terminal). Fig. \ref{figure5:2} characterizes the model of the edge cloud user.

\begin{figure}[!ht]
  \centering
  \includegraphics[scale=0.65]{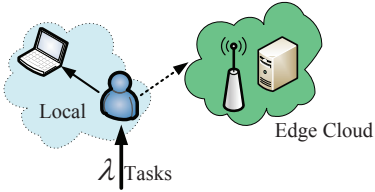}
  \caption{Edge cloud user model}
  \label{figure5:2}
\end{figure}

Minimizing the expected response time of the task has become an inevitable problem in the mobile edge computing due to the time-sensitive requirement of the task. Notably, saving energy is also extremely significant in the scenarios with constrained energy provisioning. Nevertheless, it is no longer the primary goal in numerous edge computing occasions, e.g., pilotless automobile, speech recognition, argument reality and so on, even it may be negligible in these cases. Accordingly, for simplification, we do not take energy saving into account in this study.

Moreover, a user terminal with a certain computing capability can be regarded as a small server. Hence, we also employ the queuing model to characterize it. Let $\hat{\lambda}_i$ denote the task arrival rate of $MU_i$ in a time interval, and $\hat{\mu_i}$ is her average processing rate for the task. Given that the quantity of tasks assigned to ECNs by $MU_i$ is $\lambda_i^{e}$, then  the number of tasks remaining at her own terminal is denoted by $\lambda_i^{l}$, which is written by

\begin{equation}\label{equ5:3}
 \lambda_i^{l}=\hat{\lambda}_i-\lambda_i^{e},
\end{equation}
where
\begin{equation}\label{equ5:4}
  \lambda_i^{e} \in [0, \hat{\lambda}_i].
\end{equation}

According to Eq. (\ref{equ5:1}), the average response time of tasks deployed on local terminal for $MU_i$ is given as
\begin{equation}
  T_i^{l} = \frac{1}{{{\hat{\mu _i}} - {\lambda_i^{l}}}}.
\label{equ5:5}
\end{equation}

While $MU_i$'s tasks are deployed on ENCs for execution, the total time overhead for these tasks mainly consists of the processing time $T_i^{p}$ and the transmitting time $T_i^{t}$, and is given as
\begin{equation}
  T_i^{e} = T_i^{p}+T_i^{t}.
\label{equ5:6}
\end{equation}

Here, the returned time of computing results is omitted. In our models, the mobile user, as an intelligent and selfish individual, is always greedy to search the computing nodes, which can promise the maximization of her own interests, and then allocates her tasks to them such that her interests are maximized. That is, minimizing the average response time of the task means the maximization of user interests in our work.

Furthermore, Each mobile user should be aware of the fact that the ECNs are serving other mobile users, while making decisions of task offloading. Hence, the interests of each mobile user rely on the strategies of others. It is reasonable that using the game theory to model the competition among the mobile users, which can help obtain an in-depth analytical understanding of the service provisioning problem of ENCs.

\section{Non-cooperative game for task offloading}

From aforementioned discussions, the mobile user can be viewed as the player in the non-cooperative game, and selfishly makes strategy profiles to maximize her interest. Following the routine, the following problem is how to formulate an non-cooperative game, which can quickly reach a stable situation where the task offloading strategy of each user is optimal and no one wants to change it.

\subsection{Game formulation for task offloading }
In the edge computing situation, a mobile user can distribute her tasks to any closed computing nodes including ECNs and her own terminal. For representation convenience, throughout the paper, the edge cloud node and user local terminal are collectively called computing nodes. The set of available computing nodes for $MU_i$ is represented by $\mathcal{H}_i=\mathcal{M}\cup\{MU_i\}$.

In our models, the mobile users are rational, and act independently in a selfish manner that maximizes their own interests. Let $\rho_i$ denote a probability profile determining $MU_i$'s task allocation. To minimize the task response time, $MU_i$ will do her best to achieve an optimal $\rho_i$.
$$ \rho_i=\{ {\rho_{i1}},{\cdots },{\rho_{ij}},{\cdots },{\rho_{iH}}\},$$
where~$\rho_{ij}$~denotes the probability by that $MU_i$ assigns her tasks to computing node $j$, and~$\sum\nolimits_{j=1}^H {{\rho _{ij}}}=1,H=m+1$. The strategy profile set of $MU_i$ is given as

\begin{equation} \label{equ5:7}
\mathcal{Q}_i = \left\{ {\rho} _i \Bigg| \sum\limits_{j = 1}^H {\rho _{ij}}  =1 \;\hbox{and} \;  \rho _{ij} \ge 0 \right\}.
\end{equation}

While having determined the task allocation probability vector $\rho_i$, $MU_i$ will assign her tasks to computing nodes in terms of it. In addition, the available computing capability of each computing node for $MU_i$ is represented as follows:
  $${S_i}=\{ \mu _1^i,\cdots,\mu _j^i,\cdots,\mu _H^i\},$$
where~$\mu _j^i $~represents the processing rate that computing node $j$ can provision to $MU_i$.

Let the strategy profile of other mobile users is $\rho_{-i}=\{\rho_1,\cdots,\rho_{i-1},\rho_{i+1},\cdots,\rho_n\}$, and the time of transmitting a task from $MU_i$ to computing node $i$ is denoted by $\mathcal{L}_{ij}$ (if the destination is herself, i.e., $i=j$, then $\mathcal{L}_{ij}=0$). Combining Eqs (\ref{equ5:1}) and (\ref{equ5:6}), the utility function of $MU_i$, i.e., her overall excepted response time of the task, can be written as
\begin{equation}
{T_i}(\rho_i,\rho_{-i} ) = \sum\limits_{j = 1}^H {\rho _{ij}\left(\frac{1}{{\mu
_j^i - {\rho _{ij}}{\hat{\lambda} _i}}}+\rho_{ij}\hat{\lambda}_i \mathcal{L}_{ij}\right)}
\label{equ5:8}
\end{equation}

Obviously, the mobile user, as a selfish individual, is bound to seek an optimal strategy to maximize her interests. Here, the optimal strategy of $MU_i$ is defined as follows
\begin{definition}\label{definition5:1}
(\textbf{Optimal Strategy}) Given the strategy profiles of other mobile users $\rho_{-i}$, $MU_i$'s optimal strategy is $\rho_i^* \in \mathcal{Q}_i$, if she prefers $\rho_i^*$ to any other strategy $\rho_i \in \mathcal{Q}_i$. That is
\begin{equation}\label{equ5:9}
  T_i(\rho_i^*,\rho_{-i})\leq T_i(\rho_i,\rho_{-i}),\; i = 1,\cdots,n.
\end{equation}
\end{definition}

Generally, a non-cooperative game comprises of the set of players, the strategy of the player and the set of players' strategies\cite{Myerson2013}. In our models, the mobile user is the player in the game, the player collection is denoted by $\mathcal{N}$, the strategy set of each player $MU_i$ is given by $\mathcal{Q}_i$, and the joined strategy set of all players is represented by $\mathcal{Q} = {\mathcal{Q}_1} \times  \cdots  \times {\mathcal{Q}_n}$.

Formally, we characterize the above game as a 2-tuple $G = \left\langle {\mathcal{Q},{T}} \right\rangle$, where ${T} = \left( {{T_1}, \cdots ,{T_n}} \right)$. Given the joined strategies of other participants ${\rho}_{-i}$, the objective of $MU_i$ is to achieve the optimal strategy ${\rho}_i^* \in \mathcal{Q}_i$, which minimizes the excepted response time of her tasks, ${T_i}({{\rho} _i},{{\rho} _{ - i}})$. That is

\begin{equation}\label{equ5:10}
{\rho} _i^\ast \in \mathop {\arg \min }\limits_{{{\rho} _i} \in {\mathcal{Q}_i}} {T_i}({{\rho} _i},{\rho} _{ - i}), \ ({{\rho} _i},{\rho} _{ - i}) \in \mathcal{Q}.
\end{equation}

The Nash equilibrium has a beneficial self-stability property, such that all selfish players at equilibrium can achieve a mutually satisfactory solution and no one has the incentive to change anymore. Therefore, given that the jointed optimal strategies of other participants $\rho_{-i}^*$, we have following inequality for the mobile user $MU_i$ at the Nash equilibrium.

\begin{equation}\label{equ5:11}
  T_i(\rho_i^*,\rho_{-i}^*)\leq T_i(\rho_i,\rho_{-i}^*),\; i = 1,\cdots,n.
\end{equation}

Obviously, the existence of the Nash equilibrium in a non-cooperative game is an essential condition ensuring the system stability. Hence, more attentions should be placed on it while designing a distributed task offloading mechanism with the game theory.

\subsection{ Analysis of the Nash equilibrium existence}
Following aforementioned discussions, the existence of the Nash equilibrium implies that the approach proposed by us is feasible. Hence, we will discuss and analyse the Nash equilibrium existence for the non-cooperative game $G = \left\langle {Q,{T}} \right\rangle$ in this subsection.

\begin{lemma}
For each $MU_i$, given a convex and compact strategy set $\mathcal{Q}_i$, and the function of the task expected response time ${T_i}({{\rho} _i},{{\rho} _{ - i}})$, which is continuously differentiable on ${\rho}_i \in \mathcal{Q}_i$, the function ${T_i}({{\rho} _i},{{\rho} _{ - i}})$ is convex on ${\rho}_i \in \mathcal{Q}_i$ for any strategy profile ${\rho}_{-i}$.
\label{lemma5:1}
\end{lemma}

\begin{proof}\label{proof5:1}
Obviously, the strategy set of each users $\mathcal{Q}_i$ is convex and compact in our game model, and the function of the task expected response time ${T_i}({{\rho} _i},{{\rho} _{ - i}})$ is continuously differentiable on ${\rho}_i$. In terms of literatures\cite{chen2014autonomous,scutari2012monotone}, if the Hessian matrix of the function ${T_i}({{\rho} _i},{{\rho} _{ - i}})$ is positive semidefinite,then the above lemma would hold.

According to Eq. (\ref{equ5:8}), the first derivative of the function ${T_i}({{\rho} _i},{{\rho} _{ - i}})$ for $MU_i$'s strategy vector $\rho_i$ is given as
\begin{align}
{\nabla _{{{\rho} _i}}}{T_i}({{\rho} _i},{{\rho} _{ - i}}) & = \left[ {\frac{{\partial {T_i}({{\rho} _i},{{\rho} _{ - i}})}}{{\partial \rho _{ij}}}} \right]_{j = 1}^H   \nonumber \\
& = \left({\frac{{\partial {T_i}({{\rho} _i},{{\rho} _{ - i}})}}{{\partial \rho _{i1}}}, \ldots ,\frac{{\partial {T_i}({{\rho} _i},{{\rho} _{ - i}})}}{{\partial \rho _{iH}}}} \right)   \nonumber\\
& = \left( \frac{\mu_1^i}{(\mu_1^i-\rho_{i1}\hat{\lambda}_i)^2}+2\rho_{i1}\hat{\lambda}_i\mathcal{L}_{i1}, \right. \nonumber\\
&~~~~\;~~~~\left. \ldots ,\frac{\mu_H^i}{(\mu_1^i-\rho_{iH}\hat{\lambda}_i)^2}+2\rho_{iH}\hat{\lambda}_i \mathcal{L}_{iH} \right ),\nonumber
\end{align}
and its Hessian matrix is expressed as
\begin{align}
\nabla _{{{\rho} _i}}^2{T_i}({{\rho} _i},{{\rho} _{ - i}})& = {\rm{diag}}\left\{ {\left[ {\frac{{{\partial ^2}{T_i}({{\rho} _i},{{\rho} _{ - i}})}}{{\partial (\rho _{ij})^2}}} \right]_{j = 1}^H} \right\}  \nonumber \\
& = {\rm{diag}}\left\{ {\left[ {\frac{2\mu_j^i\hat{\lambda}_i}{(\mu_j^i-\rho_{ij}\hat{\lambda}_i)^3}
}+2\hat{\lambda}_i \mathcal{L}_{ij} \right]_{j = 1}^H} \right\}.  \label{equ5:12}
\end{align}

We learn that inequality $\mu_j^i > \rho_{ij}\hat{\lambda}_i$ must hold in terms of  Expression (\ref{equ5:2}); otherwise, the queue will build up to infinity and the expected response time for the task will be infinite. Accordingly, the diagonal matrix in Eq. (\ref{equ5:12}) has all diagonal elements being positive. Hence, the Hessian matrix of the function ${T_i}({{\rho} _i},{{\rho} _{ - i}})$ is positive semidefinite and the result follows.

This remark completes the proof.
\end{proof}

\begin{theorem}\label{theorem5:1}
Non-cooperative game $G= \left\langle {Q,{T}} \right\rangle$ exists a Nash equilibrium at least.
\end{theorem}
\begin{proof}
In terms of the results of literatures \cite{Debreu1952A,fadlullah2014gtes}, an non-cooperative game existing a Nash equilibrium must satisfy two conditions: first, the strategy set $\mathcal{Q}_i$ for each user $MU_i$ is an non-empty convex closed and upper-bounded subset on Euclidean space; second, given the joined strategies of other users ${\rho}_{-i}$, the utility function of $MU_i$, i.e., the function of the task expected response time ${T_i}({{\rho} _i},{{\rho} _{ - i}})$ is continuously differentiable and convex for any strategy  ${\rho}_i \in \mathcal{Q}_i$. Obviously, the two conditions are satisfied in the game $G= \left\langle {Q,{T}} \right\rangle$ according to lemma \ref{lemma5:1} and the result follows.

this remark completes the proof.
\end{proof}

\section{Algorithm design}
In this section, we focus on the implement of the distributed task offloading approach for the non-cooperative game $G$ formulated above. Here, a decentralized task offloading frame promising the maximization of the user utility is proposed.

\subsection{Optimal task offloading algorithm for users }
Before the system arrives a Nash equilibrium, every step of improvement update carried by the participant in the game $G$ is to minimize the task expected response time, such that her utility reaches maximization. That is, the best response of each participant is the solution for the following optimization problem (labeled P1).

\begin{equation}\label{equ5:13}
  \text{P1~~~~minimize}\;\;{T_i}(\rho_i,~\rho_{-i} ), \; i = 1,\cdots,n.
\end{equation}

However, it is significantly difficult to directly address the above problem. In fact, while making the optimal decision for the task offloading at each step of improvement update, every participant considers only the computing resources being available for herself, but the affection yielded by others' strategies on her.

Hence, the optimization problem for each participant may be viewed as a centralized decision problem. Taking advantage of Nash bargaining solution (NBS) addressing such problems\cite{Boyang2017cooperation,Subrata2010Cooperative,Wu2014Cooperative}, we will adopt the NBS approach to solve it in our work.

Generally, each user has a maximum toleration for the average response time of the task. Let $\mathcal{T}_i^{max}$ denote it for $MU_i$. Therefore, the following constraint for any user must be met when she conducts task allocation on a computing node.

\begin{equation}\label{equ5:14}
  {\frac{1}{{\mu_j^i - \lambda_{ij}^{max}}}+\lambda_{ij}^{max} \mathcal{L}_{ij}}\leq \mathcal{T}_i^{max},
\end{equation}
where $\lambda_{ij}^{max}$ is the maximum rate with that $MU_i$ can assign her tasks to computing node $j$ under given constraints, such as $\mathcal{T}_i^{max}$, $\mu_j^i$ and $\mathcal{L}_{ij}$. The available processing rate $\mu_j^i$ can be obtained by observing the current state of computing node $j$ by $MU_i$. By several algebraic calculation, We have the following results in terms of the above inequality.

\begin{equation}\label{equ5:15}
  \lambda_{ij}^{max}=\frac{\mathcal{T}_i^{max}+\mathcal{L}_{ij}\mu_j^i \pm \omega}{2\mathcal{L}_{ij}},
  %\mu_j^i-\frac{1}{\mathcal{T}_i^{max}-\mathcal{L}_{ij}}.
\end{equation}
where $$\omega=\sqrt{(\mathcal{T}_i^{max}+\mathcal{L}_{ij}\mu_j^i)^2 -4\mathcal{L}_{ij}(\mathcal{T}_i^{max} \mu_j^i -1)}.$$

Notably, the following constraints must be satisfied.
\begin{eqnarray}
\mu_j^i > \lambda_{ij}^{max} \geq 0, \label{equ5:16}\\
\mathcal{T}_i^{max}-\lambda_{ij}^{max} \mathcal{L}_{ij}>0. \label{equ5:17}
\end{eqnarray}

Otherwise, the system will be impractical. Constraint (\ref{equ5:16}) guarantees the non-negative of the performance provided to $MU_i$ by computing node $j$. Expression (\ref{equ5:17}) implies that the transmitting delay can not exceed to the maximum value being acceptable for the user. Therefore, Equation (\ref{equ5:15}) can be rewritten as
\begin{equation}\label{equ5:18}
  \lambda_{ij}^{max}=\frac{\mathcal{T}_i^{max}+\mathcal{L}_{ij}\mu_j^i - \omega}{2\mathcal{L}_{ij}}.
\end{equation}

If either of expressions (\ref{equ5:16}) and (\ref{equ5:17}) cannot be satisfied on computing node $j$, then the maximum rate $\lambda_{ij}^{max}$ with that $MU_i$ can assign tasks to this computing node should be set to 0.

For applying the NBS method in our game model, the initial processing rate on computing node $j$ for $MU_i$ is denoted by $\mu_{ij}^0$, which is the maximum processing rate that computing node $j$ can provision to $MU_i$ without violating constraints (\ref{equ5:16}) and (\ref{equ5:17}). Obviously, there is $\mu_{ij}^0=\lambda_{ij}^{max}$, those computing nodes with initial processing rate 0 will be removed. Moreover, because the initial processing rate has considered the maximum transmitting delay of tasks, the traffic affect on the task response time is omitted to reduce the difficulty of our problem in the following NBS method.

According to the above propositions and the NBS definition \cite{Yaiche2000}, formally, problem P1 can be converted into the following optimization problem (labeled P2).
\begin{equation}\label{equ5:19}
  \text{P2~~~~minimize}\;\; -\sum_{j=1}^H\ln(\mu_{ij}^0-\rho_{ij}\hat{\lambda}_i),\; i = 1,\cdots,n,
\end{equation}
subject to：
\begin{eqnarray}
\sum\nolimits_{j = 1}^{H} {{\rho _{ij}}}  = 1,\label{equ5:20}\\
{\rho _{ij}} \ge 0 ,\; j = 1,\cdots,H, \label{equ5:21}\\
{{\rho _{ij}}\hat{\lambda} _i}  \leq {\mu_{ij}^0},\; j = 1,\cdots,H.
\label{equ5:22}
\end{eqnarray}

As the utility function of the mobile user is reset by Expression (\ref{equ5:19}), let $T'$ denote their new utility function collection, the non-cooperative game $G$ can be rewritten as $G'=\langle\mathcal{Q},T'\rangle$.

\begin{theorem}\label{theorem5:2}
Non-cooperative game $G'= \left\langle {Q,{T'}} \right\rangle$ exists a Nash equilibrium at least.
\end{theorem}
\begin{proof}
Through observing the formulation of P2, we can learn that $\mathcal{Q}_i$, as the strategy set of $MU_i$, is a non-empty convex closed and upper-bounded subset on Euclidean space, and the objective function Eq. (\ref{equ5:19}) is also continuously differentiable and convex for any strategy  ${\rho}_i \in \mathcal{Q}_i$. Therefore, similar to the proof of Theorem \ref{theorem5:1}, we can also assert that the game $G'$ has at least a Nash equilibrium.

This remark completes the proof.
\end{proof}

Given $\mu_{ij}^0=\lambda_{ij}^{max},\; \forall j \in \{1,\cdots,H\}$ and sorting all available computing nodes for $MU_i$ in descending order of their initial performance $\mu _{i1}^0 \ge \mu _{i2}^0 \ge \cdots \ge \mu _{iH}^0$, we have the following conclusion.
\begin{theorem}\label{theorem5:3}
the solution $\rho_i=\{ {\rho_{i1}},{\cdots },{\rho_{ij}},{\cdots },{\rho_{iH}}\}$  of optimization problem P2 is given by
 \begin{equation}\label{equ5:23}
{\rho_{ij}} = \left\{ {\begin{array}{*{10}{l}}
{\dfrac{\mu_{ij}^0}{\hat{\lambda}_i} - \dfrac{\sum\nolimits_{j = 1}^k {\mu_{ij}^0} -\hat{\lambda}_i}{k\hat{\lambda}_i},}&{j = 1,\cdots,k};\\
{0,}&{j = k + 1,\cdots,H}.
\end{array}} \right.
 \end{equation}
where $k \in \{1,\cdots,H\}$ is the largest number exactly satisfying the following expression.
\begin{equation}\label{equ5:24}
  \mu _{ik}^0 > \frac{\sum\nolimits_{j = 1}^k {\mu_{ij}^0} -\hat{\lambda}_i}{k}.
\end{equation}

\end{theorem}
\begin{proof}
In our models, assume that the total arrival rate of tasks generally does not exceed the total processing rate of the systems. Hence, the constraint (\ref{equ5:22}) is omitted in our problem.

Let $\theta  \geq 0,{\eta _j} \geq 0,\;j \in \{1,2,\cdots,H\}$  denote the Lagrange multipliers. The Lagrangian of problem P2 is given as follows:
  \begin{equation}\label{equ5:25}
 \begin{split}
    & L({\rho_{i1}},{\cdots },{\rho_{iH}},\theta ,{\eta _1},{\cdots },{\eta _H})  \\
      & = -\sum_{j=1}^H\ln(\mu_{ij}^0-\rho_{ij}\hat{\lambda}_i) - \theta (\sum\limits_{j = 1}^H {\rho_{ij}}  - 1 ) - \sum\limits_{j = 1}^H {{\eta _j}\rho_{ij}}.
 \end{split}
 \end{equation}

The first-order Kuhn-Tucker conditions and constraints are given as follows:
\begin{equation}\label{equ5:26}
\frac{\partial L}{\partial {\rho _{ij}}} = \frac{\hat{\lambda} _i}{\mu_{ij}^0 - {\rho _{ij}}\hat{\lambda} _i} - \theta  - {\eta _j} = 0,
\end{equation}
subject to the constraints
\begin{eqnarray}
  \frac{\partial L}{\partial \theta} = \sum\nolimits_{j = 1}^H {{\rho _{ij}}}  - 1 = 0, \label{equ5:27}\\
  {\eta _j}{\rho _{ij}} = 0,\;\;{\eta _j} \ge {\rho _{ij}} \ge 0,\; j = 1,\cdots,H. \label{equ5:28}
\end{eqnarray}

In terms of Eq. (\ref{equ5:28}), if $\rho_{ij}=0$, we have ${\eta _j} \geq 0$; otherwise ${\eta _j} = 0$ . Therefore, combining the Equations (\ref{equ5:26}) and (\ref{equ5:28}), we have
 \begin{equation}\label{equ5:29}
\frac{\hat{\lambda} _i}{\mu_{ij}^0 - {\rho _{ij}}\hat{\lambda} _i}\left\{ {\begin{array}{*{20}{c}}
{ = \theta ,}&{\rm{if}\;\rho _{ij} > 0};\\
{ \ge \theta ,}&{\rm{if}\;\rho _{ij} = 0}.
\end{array}} \right.
 \end{equation}

This implies that those computing nodes with lower performance can be out of the service for $MU_i$, because no tasks are assigned to them. Hence, only $\rho _{ij}>0$ needs to be considered under this case. Let $k$ is an integer promising $\rho _{ij}>0$ for all $j\in \{1,\cdots,k\}$. The Kuhn-Tucker conditions is rewritten as:
\begin{equation}\label{equ5:30}
\frac{\hat{\lambda} _i}{\mu_{ij}^0 - {\rho _{ij}}\hat{\lambda} _i} = \theta ,\;j = 1,\cdots,k,
\end{equation}
subject to
\begin{equation}\label{equ5:31}
\sum\limits_{j = 1}^k {\rho _{ij}}  = 1.
\end{equation}

By adding the equation given by Eq. (\ref{equ5:30}) for all $j\in \{1,\cdots,k\}$ and some algebraic calculation, we have
\begin{equation}\label{equ5:32}
\theta=\frac{k\hat{\lambda}_i}{\sum\nolimits_{j = 1}^k {\mu _{ij}^0}-\hat{\lambda}_i}.
\end{equation}

Using the above result in Eq. (\ref{equ5:30}), we can obtain
\begin{equation}\label{equ5:33}
\rho _{ij} =\frac{\mu_{ij}^0}{\hat{\lambda}_i} - \frac{\sum\nolimits_{j = 1}^k {\mu_{ij}^0} -\hat{\lambda}_i}{k\hat{\lambda}_i},\;j = 1,\cdots,k.
\end{equation}

If we search \emph{k} from \emph{H} to 1, it is easy to determine the largest number \emph{k} satisfying the condition $$\mu_{ik}^0 > \frac{{\sum\nolimits_{j = 1}^k {\mu_{ij}^0 - \hat{\lambda}_i} }}{k},$$ which guarantees $\rho_{ij}>0$ for all $j\in \{1,\cdots,k\}$.

This remark completes the proof.
\end{proof}

Following the result of Theorem \ref{theorem5:3}, the Optimal task offloading method for $MU_i$ is proposed in Algorithm \ref{aglorithm5:1} (called \emph{OTOM}).
In this algorithm, firstly the initial processing performance of each computing node is determined, and then sorted in descending order (lines 2-3).  Appropriate computing nodes for $MU_i$ are picked to deploy tasks such that her utility can achieve maximization (lines 4-8). It is obvious that those computing nodes with lower performance will be removed from the service for $MU_i$, once the algorithm is finished. Final, each selected computing node is set to a task allocation probability, others are set to 0 (lines 9-11). As for the complexity of Algorithm \ref{aglorithm5:1}, its overhead primarily focuses on the computation of sorting for computing nodes and determining the computing nodes participating in the service with dual loops. Hence, the computation complexity of Algorithm \ref{aglorithm5:1} is $O(m\log(m))$.

\begin{algorithm}[!ht]  %\footnotesize
\caption{OTOM($S_i,L_i,\mathcal{T}_i^{max},\hat{\lambda_i}$): Optimal task offloading method for $MU_i$ }
\label{aglorithm5:1}
\begin{algorithmic}[1]% [1] set the number code for each line
\REQUIRE ~~\\       %input data
The available processing rate for $MU_i$ on each computing node:~$S_i =\{ {\mu_1^i},\cdots,{\mu_H^i}\} $.\\
The mean delay of $MU_i$ transmitting tasks to each computing node: $\mathcal{L}_i=\{\mathcal{L}_{i1},\cdots,\mathcal{L}_{_iH} \}$.\\
The acceptable maximum response time for $MU_i$:~$\mathcal{T}_i^{max}$. \\
The mean task arrival rate on $MU_i$: $\hat{\lambda}_i$.
\ENSURE ~~\\           %output data
$MU_i$' task allocation vector: $\rho_i=\{\rho_{i1},\cdots, \rho_{iH} \}$.%
\STATE $\rho_i \leftarrow \{0,\cdots,0\}$, initialize $MU_i$' task allocation vector
\STATE Calculate the initial performance of each computing node according to Eq. (\ref{equ5:18}), $S\leftarrow\{\mu_{i1}^0,\cdots,\mu_{iH}^0\}$
\STATE Sort~$S$~in descending order
\STATE $\upsilon  \leftarrow \dfrac{{\sum\nolimits_{j = 1}^{|S|} {\mu_{ij}^0 - \hat{\lambda}_i } }}{{|S|}}$
\WHILE {$\upsilon  > \mu _{i|S|}^0$}
\STATE $S \leftarrow S\backslash \{ {\mu _{i|S|}^0}\}$
\STATE $\upsilon  \leftarrow \dfrac{{\sum\nolimits_{j = 1}^{|S|} {\mu_{ij}^0 - \hat{\lambda}_i } }}{{|S|}}$
\ENDWHILE
\FOR{$j\leftarrow 1~to~|S|$}
\STATE ${\rho _{ij}} \leftarrow \dfrac{\mu_{ij}^0 - \upsilon}{\hat{\lambda}_i} $
\ENDFOR
\RETURN {$\rho_i$}
\end{algorithmic}
\end{algorithm}

\subsection{Distributed task offloading algorithm}
In this subsection, we proposed a distributed task offloading algorithm (called \emph{DITOA}) for non-cooperative game $G'$, which are elaborated in Algorithm \ref{aglorithm5:2}. The players (the mobile users) in the game act asynchronously. First, the expected task response time and allocation probability for all players are preset to 0 (lines 1-2). Each player employs \emph{OTOM} algorithm to update her own optimal strategy in round-robin manner (lines 5-10). Specially, to recalculate the task offloading strategy, we reduce the number of tasks that were assigned to each computing node in the previous round (line 7). Moreover, a deviation $\xi$, as the convergence condition, is built to measure whether the system has reached a Nash equilibrium. That is, the algorithm gets through when the cumulative deviation difference between two adjacent iterations is less than $\xi$ (line 16). While having completed the strategy update, the mobile user computes the cumulative deviation and then broadcasts it to other players (lines 11-13). In just doing so, other players can decide whether to continue updating their task offloading strategies based on the advertised information in practice.

\begin{algorithm}[!ht] %\wuhao %\footnotesize
\caption{DITOA($\xi$): Distributed task offloading algorithm}
\label{aglorithm5:2}
\begin{algorithmic}[1]% [1] set the number code for each line
\REQUIRE ~~\\       %input data
Acceptable deviation: ~$\xi$\\
\STATE $T_i \leftarrow 0, \forall i \in \{1,\cdots,n\}$
\STATE $\rho_i \leftarrow \{0,\cdots,0\}, \forall i \in \{1,\cdots,n\}$
\REPEAT
\STATE $sum\leftarrow 0$
\FOR {each $MU_i$}
\STATE Obtain free service rate $\{ {\mu_1^i},\cdots,{\mu_H^i}\} $ and transmitting delay $\mathcal{L}_i \leftarrow \{\mathcal{L}_{i1},\cdots,\mathcal{L}_{_iH} \}$ by inspecting the available computing nodes for $MU_i$
\STATE $\mu_j^i \leftarrow \mu_j^i + {\rho_{ij}\hat{\lambda}_i}, \forall j \in \{1,\cdots,H\}$
\STATE $S_i \leftarrow \{{\mu_1^i},\cdots,{\mu_H^i}\}$
\STATE $\rho_i \leftarrow \text{OTOMU}(S_i,\mathcal{L}_i,\mathcal{T}_i^{max},\hat{\lambda}_i)$
\STATE Update the task offloading strategy by $\rho_i$
\STATE Calculate response time~$T_i$' using Eq. (\ref{equ5:8})
\STATE $sum \leftarrow sum + |T_i'-T_i|$
\STATE Advertise cumulative deviation~$sum$ to all players
\STATE $T_i \leftarrow T_i'$
\ENDFOR
\UNTIL{$sum \leq \xi$}
\end{algorithmic}
\end{algorithm}

The above algorithm can be started periodically or when system parameters change. For example, when the task arrival rate on a mobile user changes, she runs immediately the algorithm to update her task offloading strategy, and then broadcast the deviation to others. Their strategies will continue to be updated until the Nash equilibrium is reached. Besides, the processing rate on the edge cloud nodes also need to be reported to mobile users in our methods. The measurement of the processing rate can refer to the literatures\cite{Cao2013a,Mei2015,Anand1999ELISA}. In addition, given that the mobile users employ \emph{OTOM} to calculate their optimal strategies in Algorithm \ref{aglorithm5:2}, the computation complexity of \emph{DITOA} is $O(nm\log(m))$. After several iterations, the algorithm will converge to an efficient equilibrium, and then terminates.

\section{Experimental Results}
In the section, we conduct extensive simulation experiments to measure the performance of the approaches proposed by us, which include the expected task response time, the convergence speed of the proposed algorithms and so on.

For comparison purposes, we also implemented two alternative task allocation methods: proportional-scheme algorithm (called PS) proposed by the literature\cite{Chow1979Models}; global optimal scheme (denoted as GOS) employed by the literatures\cite{Li2012Optimal,Cao2013a,Huang2017Energy}. Our method is labeled DITOA.

The PS algorithm runs in such a manner that assigns tasks to each computing node in proportion to its processing rate. Hence, the faster the processing rate of the computing node is, the greater the probability of assigning tasks to it is. Allocating tasks to computing node $j$ with PS algorithm, $MU_i$ can employ following expression to calculate the task allocation probability.
\begin{equation}\label{equ5:34}
  \rho_{ij}=\frac{\mu_{ij}^0}{\sum_{j=1}^{H}{\mu_{ij}^0}}.
\end{equation}

With respect to the global optimal scheme, it is a centralized method, and its objective is to minimize the expected response time for tasks in our experiments. However, taking the traffic delay of tasks into account, it is significantly difficult to achieve the optimal solution with GOS in our problem. For simplification, the transmitting delay is omitted when we perform our experiments using GOS.

\subsection{Configurations}
In our simulation experiments, assume that there are six edge cloud nodes in a fixed physical field. The average processing rate of each edge cloud node follows an exponential distribution, which is elaborated in Table \ref{table5:2}.

\begin{table}[htbp]
%\scriptsize %font with scriptsize
\centering  % 表居中
\caption{The average processing rate of edge cloud nodes} \label{table5:2}
\begin{tabular}{>{\centering\arraybackslash}m{80pt} >{\centering\arraybackslash}m{15pt} >{\centering\arraybackslash}m{15pt} >{\centering\arraybackslash}m{15pt} >{\centering\arraybackslash}m{15pt} >{\centering\arraybackslash}m{15pt} >{\centering\arraybackslash}m{15pt}}
\toprule[1.0pt]
Edge cloud nodes &1 &2 &3 & 4&5&6\\
\hline
Average processing rate (tasks/s) &80 &60 & 100 &160 &90&70\\
\toprule[1.0pt]
\end{tabular}
\end{table}

For the task transmitting delay $L_{ij}$, we configure it as a random number distributed evenly in the interval [5ms, 100ms]. The maximum acceptable response time $\mathcal{T}_i^{max}$ for $MU_i$, $i \in {1,\cdots,n}$, is randomly generated in the interval [200ms, 280ms]. Additionally, the number of mobile users in our experiments is set to 20, the task arrival rate of each user equals to the product of her scaling factor $f_i$ and the aggregated task arrival rate $\lambda$. For example, the task arrival rate of $MU_i$ can be written as $${\lambda _i} = \lambda {f_i}.$$ The scaling factor $f_i$, $i= 1,\cdots,n$, and the local processing rate for each mobile user are given  in Table \ref{table5:3}.
\begin{table}[htbp]
%\scriptsize
\centering  % 表居中
\caption{Scaling factors $f_i$ and the local processing rate for each user} \label{table5:3}
\begin{tabular}{>{\centering\arraybackslash}m{80pt} >{\centering\arraybackslash}m{15pt} >{\centering\arraybackslash}m{15pt} >{\centering\arraybackslash}m{20pt} >{\centering\arraybackslash}m{20pt} >{\centering\arraybackslash}m{20pt}}
\toprule[1.0pt]
Mobile users&1-3 &4-7 &8-12 &13-17 &18-20 \\
\hline % \hline 在此行下面画一横线
$f_i$ &0.02 &0.035 &0.04 &0.06 &0.1 \\
Average processing rate (tasks/s) &3.6 &4.8 &5.2 &6 &7 \\
\bottomrule[1.0pt]
\end{tabular}
\end{table}

\subsection{Convergence estimation}
Since the game method proposed by us needs multiple iterations to reach the Nash equilibrium, the convergence rate of the algorithm becomes a significant metric estimating its performance. In this subsection, we will estimate the convergence rate of our methods on two aspects: acceptable deviate $\xi$ and the load rate of the system. The experimental parameters consist to the configuration provided above.

For comparison, two types of system initial scenarios are considered in our simulation experiments. The first is the \emph{initial\_0} in which the computing nodes are not assigned any task; the second is the \emph{initial\_P} in which the tasks have been assigned to each computing node in proportion to their processing rates.

Firstly, given an acceptable deviate $\xi=0.001$, which means the stop condition of our method and the arrival of the Nash equilibrium, we estimate the convergence rates of our algorithm on different system utilization rates, which are varied in the interval [0.1,0.7] with step 0.1. Moreover, the transmitting delay and the maximum acceptable response time for users are generated randomly in our experiments, hence, each experiment on different system utilization rates is carried out 100 times and its average results is described in the Fig. \ref{figure5:3}.

\begin{figure}[!ht]
  \centering
  \includegraphics[scale=1]{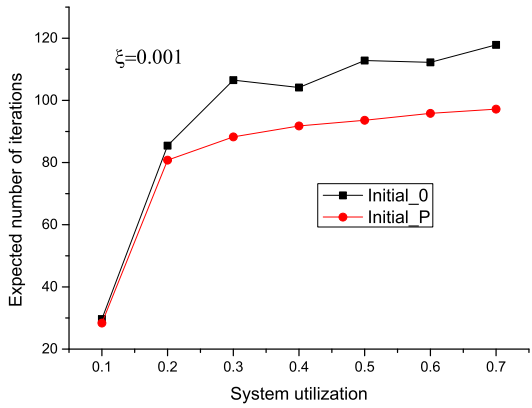}
  \caption{System utilization rate and convergence rate}
  \label{figure5:3}
\end{figure}

Fig. \ref{figure5:3} shows that given the convergence condition $\xi=0.001$, the average number of iterations of DITOA increases with the increase of the system utilization rate. The convergence rates varies evenly after the system utilization rate reaches 0.2. Furthermore, we can observe that the average number of iterations in the case of \emph{initial\_P} is lower than that of \emph{initial\_0} in Fig. \ref{figure5:3}. Generally, the task distribution obeying the case of \emph{initial\_P} is consist with the real scenario where the initial state of task allocation on computing nodes is proportional to their processing rate, so the experimental results demonstrate the effectiveness of our schemes.

Next, given the system utilization rate is set to 0.5, we estimate the convergence rate of our methods on distinct acceptable deviates. Similar to the above experiments, we run the experiments 100 times and then calculate the average number of iterations as experimental results. Fig. \ref{figure5:4} shows that the convergence rate of \emph{initial\_P} is prior to that of \emph{initial\_0}. It implies that the system in the case of \emph{initial\_P} is more close to the Nash equilibrium. Additionally, while the acceptable deviate is set to 0.01, the average number of the iteration of our algorithm is 18. The above experiment results demonstrate that the algorithm proposed by us can converge quickly in an acceptable deviate range.

\begin{figure}[!ht]
  \centering
  \includegraphics[scale=1]{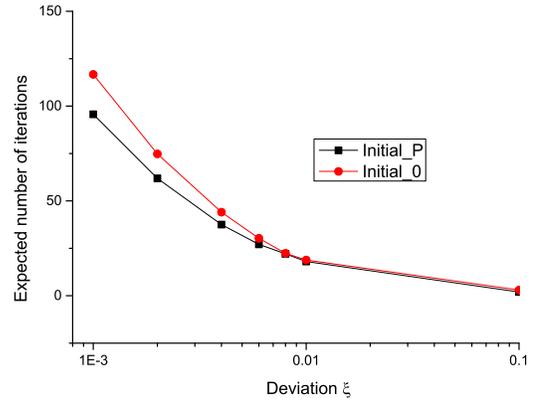}
  \caption{Acceptable deviate and convergence rate}
  \label{figure5:4}
\end{figure}

\subsection{Estimation of average response time}
In this subsection, we will estimate the average response time of user tasks by varying the utilization rate of the system. The parameters of simulation experiments are given in terms of previous experiment configurations. Considering that the transmitting delay and the acceptable deviate are set randomly, similar to the above experiments, each simulation experiment is also carried out 100 times. The convergence condition (acceptable deviate) $\xi$ is set to 0.001 in the experiments.
\begin{figure}[!ht]
  \centering
  \subfigure [DITOA]{
  %\label{fig:subfig:9a} %% 第一幅图的标签
  \includegraphics[scale=1.0]{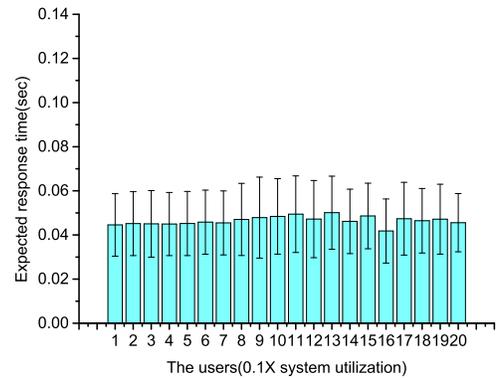}}
  %\hspace{1in}
  \subfigure [PS]{
  %\label{fig:subfig:9b} %% 第一幅图的标签
  \includegraphics[scale=1.0]{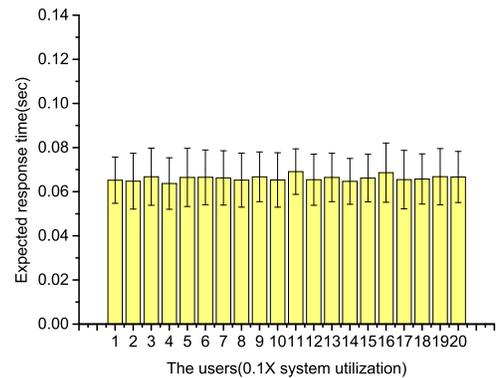}}
  \caption{~0.1X~system utilization rate and the average response time of user tasks}
  \label{figure5:5}%\label{fig:experiment6}
\end{figure}

\begin{figure}[!ht]
  \centering
  \subfigure [DITOA]{
  %\label{fig:subfig:9a} %% 第一幅图的标签
  \includegraphics[scale=1.0]{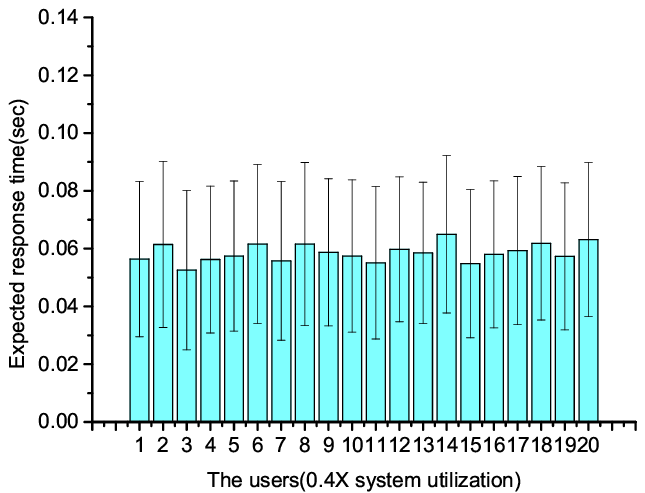}}
  %\hspace{1in}
  \subfigure [PS]{
  %\label{fig:subfig:9b} %% 第一幅图的标签
  \includegraphics[scale=1.0]{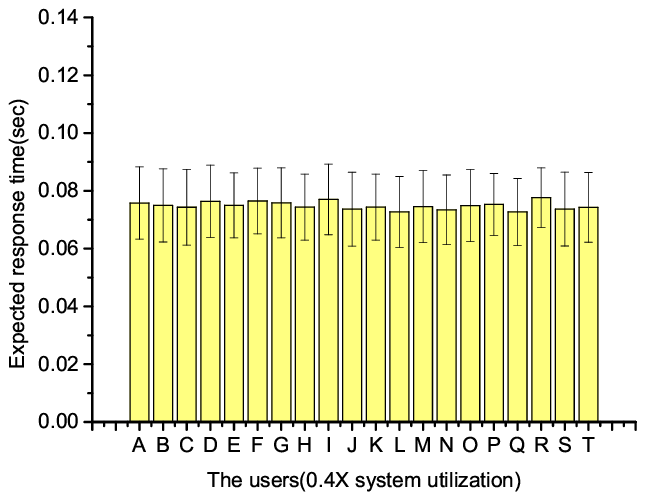}}
  \caption{~0.4X~system utilization rate and the average response time of user tasks}
  \label{figure5:6}%\label{fig:experiment6}
\end{figure}

\begin{figure}[!ht]
  \centering
  \subfigure [DITOA]{
  %\label{fig:subfig:9a} %% 第一幅图的标签
  \includegraphics[scale=1.0]{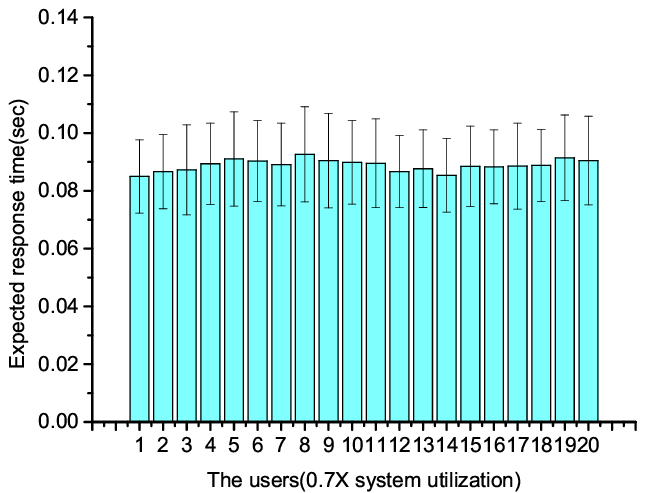}}
  %\hspace{1in}
  \subfigure [PS]{
  %\label{fig:subfig:9b} %% 第一幅图的标签
  \includegraphics[scale=1.0]{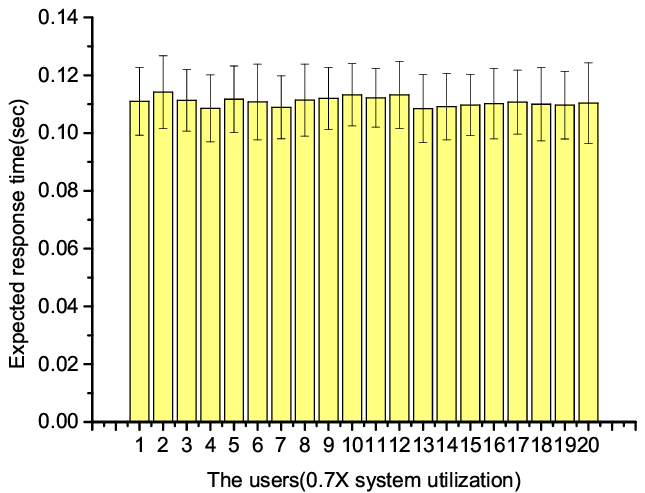}}
  \caption{~0.7X~system utilization rate and the average response time of user tasks}
  \label{figure5:7}%\label{fig:experiment6}
\end{figure}

Figs. \ref{figure5:5}~,~\ref{figure5:6}~and~\ref{figure5:7} present the average task response time with standard deviate for PS and DITOA methods, when the system utilization rates are at 0.1X, 0.4X and 0.7X respectively. The experimental results show that the expected task response time of our method is lower than that of PS method at distinct system utilization rates. The main reason is that the strategy employing by users in our method is optimal for individuals, rather than PS method is not optimal for individuals due the task allocation in proportion to the processing rates of computing nodes. Accordingly, it makes the performance of our method is superior to that of PS.

According to the previous configurations, GOS method does not need to be executed repeatedly because its transmitting delay is omitted. Furthermore, the result yielded by GOS is the optimal average response time of the whole system. Hence, GOS is not included in the above experiments. It is obvious that the task response time of each user in PS and DITOA is different. To compare PS, DITOA and GOS, we adopt the following expression to normalize their overall average response time of the tasks.
\begin{equation}\label{equ5:35}
  \frac{1}{\sum_{i=1}^{n}{\hat{\lambda}_i}}\sum_{i=1}^{n}{\hat{\lambda}_iT_i(\rho_i,\rho_{-i})}.
\end{equation}

Fig. \ref{figure5:8} describes the overall average response time of the task for the three methods at distinct system utilization rates. It can been seen that the task response time of GOS is the lowest, and PS is the highest. The primary reason is that the transmitting delay is ignored in GOS method. Specially, while the system utilization rate is low, the task processing time on the computing node is far less than the task transmitting delay in DITOA and PS. Obviously, GOS is superior to the other two methods. Moreover, as the increase of the system utilization rate, the difference between the processing time and the transmitting delay decreases, and the performance of DITOA approaches that of GOS. In some senses, GOS can be regarded as the optimization benchmark in our experiments. Therefore, the experimental results validate the effectiveness of our methods.
\begin{figure}[!ht]
  \centering
  \includegraphics[scale=1.0]{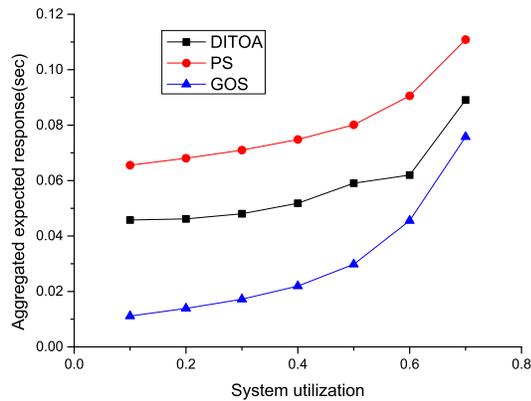}
  \caption{System utilization rate and the overall average response time of the task}
  \label{figure5:8}
\end{figure}

\section{Conclusion}
In this paper, we investigated the task offloading problem in the mobile edge cloud computing. We employ NBS approach to compute the mobile user utility, and then present a non-cooperative game framework and associated algorithm to address the problem, such that each mobile user can minimize the expected response time of tasks in a compromise scenario (approaching a Nash equilibrium point). The presented algorithm has relatively low complexity and distribution execution characteristic, and thus, can be easily implemented to improve the reliability and robustness of the system. The effectiveness of our approach was assessed by performing simulated experiments. The experimental results demonstrated that our approach could outperform alternative popular methods achieving better performance. Besides, our proposed approach can be applied to other resource allocation models.

In the future, we plan to explore the VM or container migration among multiple mobile edge clouds as an extension of our work. We are also interested in implementing a prototype allocation system in an experimental edge cloud platform to further study the performance of our proposed approach.

% or


%\appendices

%\section{}

\ifCLASSOPTIONcompsoc
  \section*{Acknowledgments}
   This work was partially supported by the National Natural Science Foundation of China (No. 61173107, 91320103, 61672215, U1613209), National High-tech R\&D Program of China (863 Program) (No. 2012AA01A301-01),the Special Project on the Integration of Industry, Education and Research of Guangdong Province, China (No.2012A090300003), Science research project of Department Education of Hunan Province, China (No.2016C0270), Social Science Foundation of Hunan Province, China (No. 16YBA050) and the Science and Technology Planning Project of Guangdong Province, China (No.2013B090700003). The correspondence author is Zhiyong Li.

 \else
  \section*{Acknowledgment}
\fi


\ifCLASSOPTIONcaptionsoff
  \newpage
\fi



\bibliographystyle{IEEEtran}
\end{document}